\documentclass[letterpaper, 10 pt, conference]{ieeeconf}

\IEEEoverridecommandlockouts                              

\usepackage{amssymb,amsmath,amsfonts,mathtools}

\usepackage{algorithmic}
\usepackage{subfig}
\usepackage[colorlinks=true]{hyperref}
\usepackage{graphicx}
\usepackage{textcomp}
\usepackage{mathtools}
\usepackage[colorlinks=true]{hyperref}
\usepackage{algorithm}

\usepackage{calrsfs}

\newtheorem{theorem}{Theorem}[section]
\newtheorem{lemma}[theorem]{Lemma}
\newtheorem{definition}[theorem]{Definition}
\newtheorem{corollary}[theorem]{Corollary}

\newtheorem{problem}[theorem]{Problem}

\newtheorem{remark}[theorem]{Remark}

\newtheorem{assumptions}[theorem]{Assumptions}

\DeclareMathAlphabet{\pazocal}{OMS}{zplm}{m}{n}

\renewcommand{\theequation}{\thesection.\arabic{equation}}
\renewcommand{\thefigure}{\thesection.\arabic{figure}}
\renewcommand{\thetable}{\thesection.\arabic{table}}

\newcommand{\sr}{\stackrel}

\newcommand{\rar}{\rightarrow}

\newcommand{\tri}{\sr{\triangle}{=}}

\newcommand{\bea}{\begin{eqnarray}}
\newcommand{\eea}{\end{eqnarray}}
\newcommand{\bes}{\begin{eqnarray*}}
\newcommand{\ees}{\end{eqnarray*}}
\newcommand{\bce}{\begin{center}}
\newcommand{\ece}{\end{center}}

\def\VR{\kern-\arraycolsep\strut\vrule &\kern-\arraycolsep}
\def\vr{\kern-\arraycolsep & \kern-\arraycolsep}

\newcommand{\ben}{\begin{enumerate}}
\newcommand{\een}{\end{enumerate}}

\newcommand{\bi}{\begin{itemize}}
\newcommand{\ei}{\end{itemize}}

\newcommand{\bp}{\begin{problem}}
\newcommand{\ep}{\end{problem}}
\newcommand{\hso}{\hspace{.1in}}
\newcommand{\hst}{\hspace{.2in}}

\newcommand{\bc}{\begin{center}}
\newcommand{\ec}{\end{center}}

\title{\LARGE \bf  Private and Common Information States in Decentralized   Parallel  Dynamic Programming  for  Delayed Sharing Patterns
}

\author{Charalambos D. Charalambous$^1$, Umarbek Guvercin$^2$,  Seddik   Djouadi$^2$ 
\thanks{$^{1}$Charalambos D. Charalambous is  with the Faculty of Electrical and Computer Engineering, University of Cyprus, Nicosia 1678, Cyprus
{\tt\small chadcha@ucy.ac.cy}}
\thanks{$^{2}$Seddik Djouadi and  Umarbek Guvercin are  with the Faculty of Electrical Engineering and Computer Science, University of Tennessee, Knoxville, TN, 37996, USA 
{\tt\small \{mdjouadi,ugvercin\}@utk.edu}}%
\thanks{}%
}

\begin{document}

\maketitle
\thispagestyle{empty}
\pagestyle{empty}

\begin{abstract}
This paper develops a dynamic programming (DP) approach  for  decentralized  stochastic optimal control problems with   delayed sharing information patterns,  which exhibits  the fundamental Properties of classical DP   of centralized  partially observable Markov decision problems (POMDPs):  the   value functions and  information states   depend on the actions of the   minimizing controls and not their  strategies. 

This is achieved by invoking  the  concept of Person-by-Person  (PbP)  optimality, in which 
  each control strategy is associated with  a value function conditioned on  its  assigned delayed sharing  information pattern, when all other strategies are fixed to their optimal responses. 
The  value functions satisfy generalized  and  simplified DP equations. These are   used to derive necessary and sufficient conditions for PbP optimality. 

The simplified DP equations are obtained by invoking  the  structural property that optimal strategies are separated and functionals of  two information  states:   
 1) a  private    \^a posteriori probability distribution  based  on the  information pattern of the strategy,   and  2)  a  centralized \^a posteriori probability  distribution based   on the shared or common information to  all strategies, each  satisfying a   Markov recursion. 
 
The DP approach of this paper, settles a long standing open problem since the appearance of  {  $T-$step delayed  sharing patterns} in \cite[Section~IV.G]{witsenhausen1971}, in terms of generalizing the fundamental properties of classical DP approach.

 
\end{abstract}

\section{INTRODUCTION}
\label{sec:int}
Witsenhausen in the 1971 seminal   paper \cite{witsenhausen1971}, introduced 
    a    general decentralized Markovian  stochastic network with    multiple  observation posts collecting information and  multiple control stations  applying control actions.  The control actions are generated by  strategies,  which are  assigned different  {\it ``Information Patterns or Structures''}.  The objective of the multiple  strategies  is to   optimize   a single payoff,  common  to all  strategies.  
    Extensively discussed in \cite{witsenhausen1971}  
 is   the {\it  $T-$step delayed  sharing patterns}  \cite[Section~IV.G]{witsenhausen1971}, when   each   
strategy   has  access  to {\it a private information component},  and a {\it common or shared    information component} received by all other  strategies   with  delay $T$ units of time.  
        
        Over the last 50 years, several   decentralized optimization methods are proposed. These are classified  into   3 directions. 

{\it  1) The Single Dynamic Programming Approach  for $T-$Step Delayed  Sharing Patterns  \cite{kurtaran-sivan1973,sandell-athans1974,kurtaran1975,yoshikawa1975,varaiya-walrand1978,kurtaran1979,bansal-basar1987} and generalizations  \cite{nayyar-mahajan-teneketzis2011,nayyar-mahajan-teneketzis2013,nayyar-teneketzis2019}. } 

 {\it  2) The Static Reduction \cite{witsenhausen1988} and generalizations  \cite{charalambous2014equivalence,teslang-djouadi-charalambous:ACC2021}.}
 
 {\it  3)  The Pontryagin's Decentralized  Stochastic Maximum Principle (SMP)  \cite{charalambous-ahmed:IEEEAC2017a,charalambous-ahmed:IEEEAC2018,charalambous-ahmed:IEEEAC2017b,charalambous:MCSS2016}.} 

These  methods were   inspired  by concepts of   {\it static team theory} 
\cite{radner1962,marschak-radner1972}
called   {\it person-by-person\footnote{PbP optimality may be viewed as a special case of the concept  of Nash-equilibrium; it is applied  in static team optimization problems with multiple  agents  to deal with asymmetry of information of  the agents.} (PbP) and  team  optimality}.
The generalization  of static team theory to  
  decentralized  stochastic dynamical systems and  networks remains to this date an active research area.  
  
In this paper we focused on 1). We develop a new DP  approach that  satisfies   the   two fundamental  Properties \# 1, \#2  of  classical DP  of  centralized   partially observable Markov decision problems (POMDPs):
 
{\bf Property \#1.} The cost-to-go, value processes  or value  functions involve conditional probability distributions, which are   independent of the  optimizing strategies;   the optimization of  the right hand side of the  DP equations is over the action spaces of the controls and not over their   strategy spaces.

{\bf Property \#2.} The DP equations are expressed in terms of information states  (i.e., satisfy  the Markov property); the information states  depend on the  actions of the control and not   their strategies,  they  are   sufficient statistics for the    strategies, and a generalized separation principle holds.

Indeed, the widespread success of    classical  centralized DP  of  POMDP  \cite{kumar-varayia:B1986,caines1988,hernandezlerma-lasserre1996} is attributed to  Properties \#1, \#2.  
In particular,  Property \#1 is essential to  develop efficient DP algorithms, while Property \#2 is fundamental    for optimal strategies to be Markovian, rather than  expand with time.

        
\subsection{Literature Review of the Single Dynamic Programming Approach  for $T-$Step Delayed  Sharing Patterns}
Subsequent to  the appearance of     \cite{witsenhausen1971}, several   studies   \cite{kurtaran-sivan1973,sandell-athans1974,kurtaran1975,yoshikawa1975,varaiya-walrand1978,kurtaran1979,bansal-basar1987,nayyar-mahajan-teneketzis2011,nayyar-mahajan-teneketzis2013,nayyar-teneketzis2019}     focused on developing a dynamic programming (DP) approach   for   decentralized Markovian stochastic networks with {\it  $T-$step delayed  sharing patterns}, by invoking Witsenhausen's separation  assertion    \cite[Assertion 8, pp.1562]{witsenhausen1971} (or variations of it).  
 This gave rise to the
    DP approach  \cite{sandell-athans1974,yoshikawa1975,varaiya-walrand1978,nayyar-mahajan-teneketzis2011,nayyar-mahajan-teneketzis2013,nayyar-teneketzis2019},  based on  {\it a single  cost-to-go} for all strategies conditioned on the  common or sharing  information component available to all strategies with delay  $T-$ units of time. To ensure the DP approach is tractable,   use is made of Assertion   8 in  \cite[Assertion 8, pp.1562]{witsenhausen1971} (or variations)   to compress the common  information component available to all strategies, using an  \^a posteriori probability distribution  conditioned on the common  information component.  
    
 \indent As we discuss below,  the single DP  approach  \cite{sandell-athans1974,yoshikawa1975,varaiya-walrand1978,nayyar-mahajan-teneketzis2011,nayyar-mahajan-teneketzis2013,nayyar-teneketzis2019}  suffers from various  technical  limitation.
    
The single DP approach uses a   {\it  single cost-to-go},  leading  to a single  DP equation, which   does not  satisfy  fundamental  Properties \# 1, \#2  of  classical DP  of  centralized   POMDPs. 
%
%
 This  is   apparent from \cite{nayyar-mahajan-teneketzis2011,nayyar-mahajan-teneketzis2013,nayyar-teneketzis2019}
 and  acknowledged in   \cite[pp.1610-1611]{nayyar-mahajan-teneketzis2011}. 

  The DP approach based on   {\it the single cost-to-go}  requires the existence of a common or shared information component received by all strategies (see 
 Witsenhausen  \cite[below Assertion 8, pp.1562]{witsenhausen1971} and   \cite[Abstract]{nayyar-teneketzis2019}):  
if  there is no  shared information component among the controls 
 the  conditioning  information of the  {\it single  cost-to-go}  is the {\it null set},  hence the
  DP approach ``fails".  
   
 \indent  The single  DP  approach 
  does not fully  characterize the optimality of the strategies. This limitation  is obvious  from Sandell's and  Athan's  \cite{sandell-athans1974} derivation  
of the solution of  LQG problems with $1-$step delayed sharing patterns (and    \cite{kurtaran-sivan1973,kurtaran1975,yoshikawa1975}). In fact,     
the single DP equation is used as   intermediate or   preliminary  step, towards obtaining the optimal strategies.   The complete derivation   involves  another decentralized   optimization problem, which is solved  using the concept of  PbP  optimality of static team theory \cite{radner1962}, and the  inherent simplicity of  $1-$step delayed sharing. Similar steps are  involved in 
 the counterexample
   \cite[Section~II]{varaiya-walrand1978}.

 \subsection{New Decentralized DP Approach and Separation Principle}
 \label{sub-DP} 
 {\bf The Approach.}
 We develop a new DP approach for  decentralized Markovian stochastic networks with  $T-$step delayed  sharing patterns, by invoking  the concept  of PbP optimality of static team theory \cite{radner1962,marschak-radner1972}. 
%
 PbP optimality (see Definition~\ref{def-pbp})   is   analogous to the notion of classical Nash-equilibrium,  with the important difference that there is  a   single payoff common to all strategies,  and the information structures of the strategies are different.   
 
Based  on the definition of  PbP optimality, 
%
we  develop  a new generalized  DP and separation principle approach that satisfies the   two fundamental  Properties \#1, \#2,  of classical  centralized POMDPs. 
%
%
%
%


{\bf  The Main Contributions.} 
The main results  are   necessary and sufficient conditions for PbP optimality, using generalized   DP equations and information states, which  are consistent with Properties \#1, \#2. 

The  information states  are sufficient statistics for the strategies, leading to different generalized and simplified  DP equations,  expressed in terms of the so-called, {\it semi-separated strategies} and  {\it separated strategies}. These are natural generalizations of classical centralized DP,  information states and separated strategies  of POMDP \cite{bertsekas-shreve1978,kumar-varayia:B1986}.

An important feature of  the  optimal {\it separated strategies}  is that,  each strategy requires  2  information states or beliefs:

1) {\it The private information state} of each control, based on its assigned  $T-$step delayed shared information pattern.

2) {\it The  centralized information state},  based on the common or shared information component received by all controls.

 The rest of the paper is structured as follows.

Section~\ref{sect:model} introduces  the  decentralized stochastic  network.

Section~\ref{discrete}   derives the main results of  this paper.

\section{The  Decentralized Stochastic  Network}
\label{sect:model}
We introduce 
 the decentralized discrete-time  stochastic network.


{\bf Notation.} ${\mathbb R} \tri (-\infty,\infty)$, ${\mathbb  Z} \tri \{\ldots, -1,0,1, \ldots\}$,    ${\mathbb Z}_+ \tri \{1,2, \ldots\}$, ${\mathbb Z}_+^n \tri \{1,2, \ldots,n\}$,  $n \in {\mathbb Z}_+$. 
  Given  $s^{(K)} \tri \{s^1, \ldots, s^K\}, K \in \mathbb{Z}_+$, we  define   $s^{-k} \tri s^{(K)} \setminus\{s^k\}=\{s^1, \ldots, s^{k-1},s^{k+1}, \ldots, s^K\}$, i.e.,  $s^{(K)}$ minus element $\{s^k\}$.  \\ 
 $\{({\mathbb X}_t,{\cal B}({\mathbb X }_t))\big| t\in\mathbb{Z}_+^n\}$ denotes  measurable spaces, where ${\mathbb X}_t$ is  confined to complete separable metric space or Polish space, and ${\cal B}({\mathbb X}_t)$ is the Borel $\sigma-$algebra of subsets of ${\mathbb X}_t,  \forall t\in\mathbb{Z}_+^n$. Points in the product space ${\mathbb X}_{1,n}\triangleq{{\prod}_{t\in\mathbb{Z}_+^n}}{\mathbb X}_t$ are denoted by $x_{1,n}\triangleq \{x_1,\ldots, x_n\} \in{\mathbb X}_{1,n}$, and their restrictions for any $(m,n)\in\mathbb{Z}_+ \times \mathbb{Z}_+$ by $x_{m,n}\triangleq \{x_{m},\ldots, x_n\} \in{\mathbb X}_{m,n}, n\geq{m}$. Hence,  ${\cal B}({\mathbb X}_{1,n})\triangleq\otimes_{t\in\mathbb{Z}_+^n}{\cal B}({\mathbb X}_t)$ denotes  the $\sigma-$algebra on ${\mathbb X}_{1,n}$ generated by cylinder sets $\{\{x_1,\ldots, x_n, \ldots,\}\in{\mathbb X}_{1,\infty}\big|x_j\in{A}_j, A_j\in{\cal B}({\mathbb X}_j),~j\in\mathbb{Z}_+^n\}$. We use the convention $X_{k,n}=X_{\max\{k,1\}, n}$ and $X_{k,n}=\{\emptyset\}, \forall k>n, (k,n)\in {\mathbb Z}\times {\mathbb Z}$. 
{\bf  The Model of  Decentralized Stochastic Network.}
We consider  a  Markovian stochastic network 
 on the  finite horizon 
 ${T}_+^{n}\tri \{1,2, \ldots, n\}$,   defined on an underlying  probability space $(\Omega, {\cal F}, {\mathbb P})$,  described  by  the following elements. 

(a) The   unobservable state process, $X_{1,n} \tri \{X_1, X_2$, $\ldots, X_n\}$, 
 $X_t: (\Omega, {\cal F}) \rar ({\mathbb X}_t, {\cal B}({\mathbb X}_t)),\hso \forall t \in {T}_+^{n}$. 
 
(b)  The observations proc. $Y_{1,n}^{(K)} \tri \{Y_{1,n}^1,   \ldots, Y_{1,n}^K\}$,  
$Y_t^k: (\Omega, {\cal F}) \rar ({\mathbb Y}_t^k,{\cal B}(    {\mathbb Y}_t^k   )), \hso \forall t \in {T}_+^{n},  \hso \forall k \in {\mathbb Z}_+^K$.

(c) The control proc.  $U_{1,n}^{(K)} \tri \{U_{1,n}^1, \ldots, U_{1,n}^K\}$,  
$U_t^k: (\Omega, {\cal F}) \rar ({\mathbb U}_t^k,{\cal B}(  {\mathbb U}_t^k   )),  \forall t \in {T}_+^{n} ,  \forall k \in {\mathbb Z}_+^K$,
where  $\big\{ {\mathbb U}_t^k | t \in {T}_+^{n}\big\}$ are referred to as the {\it  action spaces }of  $U_{1,n}^k$, $ \forall k \in {\mathbb Z}_+^K$. 

(d) The $T-$Step Delayed Sharing Patterns.  For each $(t, k)\in T_+^{n} \times {\mathbb Z}_+^K$, control  $U_t^k$ is assigned  the information  pattern $I_{t}^k $, 
\begin{align*}
& I_t^k\tri \big\{\Delta_t^{(K)}, \Lambda_t^k\big\}: (\Omega, {\cal F}) \rar ({\mathbb I}_t^k, {\cal B}({\mathbb I}_t^k)),  \forall (t, k)
\\
& 
 \Lambda_t^k: (\Omega, {\cal F}) \rar ({\mathbb L}_t^k, {\cal B}({\mathbb L}_t^k)) \;  \mbox{private comp. of control $k$,}\\
  &\Delta_t^{(K)} : (\Omega, {\cal  F}) \rar ({\mathbb D}_t, {\cal B}({\mathbb D}_t))  \; \mbox{common comp. $\forall$ controls}
\end{align*}
where the two components are specified by
\begin{align}
\Delta_t^{(K)} =   \big\{Y_{1,t-T}^{(K)}, U_{1,t-T}^{(K)}\big\},\; \Lambda_t^k =     \big\{Y_{t-T+1, t}^k, U_{t-T+1, t-1}^k \big\} \nonumber
\end{align}
for any delay $T \in {T}_+^{n}$. 
By our  convention $\Delta_t^{(K)}=\{\emptyset \}, \forall t \in {T}_+^{n}$ such that $ t <T+1$. 

(e) The strategies of the Controls.  Given,  $\{I_t^k|t \in T_+^n, k \in {\mathbb Z}_+^K\}$, 
for each $(t, k)\in T_+^{n} \times {\mathbb Z}_+^K$, control  $U_t^k$ is
generated by a  measurable functions or strategies $\gamma_t^k(\cdot)$    as follows. 
\begin{align}
U_t^k=\gamma_t^k(I_t^k)=\gamma_t^k(\Delta_t^{(K)}, \Lambda_t^k),  \;     \forall t \in {T}_+^{n},\;  \forall k \in {\mathbb Z}_+^K. \label{type_a}
\end{align}
 For each $k \in {\mathbb Z}_+^K$,  we denote by  ${\cal  U}_{1,n}^{k}$ the set of  such  admissible  strategies of the $k$th  control  $U_{1,n}^k$ over  $T_+^{n}$. \\For notational simplicity we use the definitions,
\begin{align} 
&{\cal  U}_{1,n}^{k} \tri \prod_{t=1 }^n {\cal  U}_{t}^k,\hso {\cal  U}_{1,n}^{-k} \tri \prod_{j=1,j\neq k }^K {\cal  U}_{1,n}^j, \\ 
 &\gamma_{1,n}^{-k}(\cdot)\tri   \big\{\gamma_{1,n}^{1},
 \ldots,  \gamma_{1,n}^{k-1},\gamma_{1,n}^{k+1}, \ldots, \gamma_{1,n}^{K}\big\}(\cdot)\in {\cal  U}_{1,n}^{-k}, \nonumber
 \\
& \gamma_{t}^{-k}(I_t^{-k})
 \tri   \big\{\gamma_{t}^{1}(I_t^{1}),\ldots,  \gamma_{t}^{k-1}(I_t^{k-1}),   \gamma_{t}^{k+1}(I_t^{k+1}), \ldots, \nonumber \\
 &  \ldots,  \gamma_{t}^{K}(I_t^{K})\big\}, \hso  I_t^k=\Delta_t^{(K)}\cup \Lambda_t^k,              \label{not_1a} \\
 &\equiv \gamma_t^{-k}(\Delta_t^{(K)}, \Lambda_t^{-k})\equiv \gamma_t^{-k},  \;  \forall t \in { T}_+^{n},   \;    \forall k \in {\mathbb Z}_+^K.\label{not_3}
 \end{align}

 
 Next, we introduce the various Markovian conditional probability measures (PMs). 

{\it 1)   The Conditional Probability Measure (PM)} of $X_{t+1}$ conditioned on $(X_{1,t}, Y_{1,t}^{(K)}, U_{1,t}^{(K)})$ is 
\begin{align}
& {\mathbb P} \Big\{X_{t+1} \in dx_{t+1} \big| X_{1,t},Y_{1,t}^{(K)},   U_{1,t}^{(K)}\Big\}, \;  \;   \forall t \in  T_+^{n-1}\nonumber \\
&=
 {\bf P}_{X_{t+1}|X_{t}, U_{t}^{(K)}}=S_{t+1}(dx_{t+1}|X_t, U_{t}^{(K)}).  \label{i-d1}
 \end{align}

{\it  2) The Conditional PM} of $Y_{t}^k$ conditioned on $(X_{1,t},Y_{1,t-1}^{(K)}, Y_t^{-k}, U_{1,t-1}^{(K)})$ is  
\begin{align}
 &{\mathbb P}\Big\{Y_{t}^k \in dy_t^k \big| X_{1,t},Y_{1,t-1}^{(K)},Y_t^{-k}, U_{1,t-1}^{(K)}\Big\}, \;   \forall  t\in T_+^n, k \in {\mathbb Z}_+^K \nonumber \\
 &=
 {\bf P}_{Y_{t}^k|X_{t}, U_{t-1}^{(K)}} =Q_{t}^k( dy_t^k|X_t,  U_{t-1}^{(K)})   . \label{i-d2}
 \end{align} 
From  (\ref{i-d2}), by Bayes' Theorem    we also deduce, 
\begin{align}
 &{\mathbb P} \Big\{Y_{t}^{(K)} \in dy_t^{(K)} \big| X_{1,t},Y_{1,t-1}^{(K)},  U_{1,t-1}^{(K)}\Big\}, \; \forall  t\in T_+^n  \nonumber \\
 &= Q_{t}^{(K)}(dy_t^{(K)} |X_t,  U_{t-1}^{(K)})\\
& = \prod_{k=1}^K Q_{t}^k(dy_t^k |X_t,U_{t-1}^{(K)} ). \label{i-d2-j-2}
 \end{align}

\indent {\it 3) The Conditional} PM of $U_t^{(K)}$ conditioned on $(X_{1,t},Y_{1,t}^{(K)}, U_{1,t-1}^{(K)})$ with  $\gamma^{(K)}(\cdot) \in {\cal U} _{1,n}^{(K)}$, 
\begin{align}
 &{\mathbb P} \Big\{U_{t}^{(K)} \in du_t^{(K)} \big| X_{1,t}, Y_{1,t}^{(K)}, U_{1,t-1}^{(K)}  \Big\}, \; \forall  t\in T_+^{n} \nonumber \\
  &=
 {\bf P}_{U_{t}^{(K)}|Y_{1,t}^{(K)},  U_{1,t-1}^{(K)}}= P_{t}^{(K)}(du_t^{(K)} |I_t^{(K)})     \label{i-d2-j-22} \\
& = \prod_{k=1}^K P_{t}^k(du_t^k |I_t^k)=\prod_{k=1}^K \mu_{\gamma_t^k(I_t^k)}(du_t^k )
 \label{i-d2-j-22-a}
 \end{align}
$P_{t}^k(du_t^k|I_t^k)=\mu_{\gamma_t^{k}(I_t^k)}(du_t^{k})$ is the  Dirac measure  at  $\gamma_t^{k}(I_t^k) $.

 {\it 4) The Average  Payoff  Function.}
  Given  a   $K$-tuple  $\gamma^{(K)}\tri \{\gamma^1, \ldots, \gamma^K\} \in {\cal  U}_{1,n}^{(K)} $, the average   payoff  is  
\begin{align}
&J_{n} (\gamma^{(K)}) \tri {\mathbb E}^{{\gamma^{(K)}}} \Big\{ \sum_{t=1}^{n} \ell(t, X_t, \gamma_t^{(K)}) \Big\}   \label{i8-cost_a}
\end{align}
where  $\ell(t, \cdot), \forall t$  is lower semicontinuous, bounded from below.

\indent To  deal with asymmetry of information of the  strategies, we use   the definition   of  {\it person-by-person (PbP) optimality }
 from {\it static team theory},  \cite{radner1962,marschak-radner1972}.

 \begin{definition}(Decentralized PbP Optimality) \\
 \label{def-pbp}
 The  $K-$tuple of strategies   $\gamma^{(K),o}\tri   \{\gamma^{1,o},  \ldots, \gamma^{K,o}\} \in {\cal  U}_{1,n}^{(K)}$
  is called  {\it decentralized PbP optimal}, if it satisfies 
\begin{align}
 {J}_{n}(\gamma^{k,o}, \gamma^{-k,o}) \leq  {J}_{n}(\gamma^{k}, \gamma^{-k,o}), \;
  \forall \gamma^k \in {\cal  U}_{1,n}^{k} \label{pbp}
\end{align}
 $ \forall k \in {\mathbb Z}_+^K$, where  
  \begin{align}
&J_{n}(\gamma^{k,o}, \gamma^{-k,o}) \tri  \inf_{  \gamma^{k}\in {\cal U}_{1,n}^k}   {\mathbb E}^{^{\gamma^k, \gamma^{-k,o}}} \big\{ \sum_{t=1}^{n} \ell(t, X_t,\gamma_t^k,  \gamma_t^{-k,o})\big\}\nonumber 
\end{align}
$
\big\{\gamma_t^{k}, \gamma_t^{-k,o}\big\} \tri \big\{\gamma_t^{k}(I_t^{k}),  \gamma_t^{-k,o}(I_t^{-k}) \big\}$  (see (\ref{not_1a})),  where ${J}_{n}(\gamma^{k,o}, \gamma^{-k,o})$ is the optimal payoff of  strategy $\gamma^k \in {\cal  U}_{1,n}^{k}$, corresponding to  fixed optimal responses  $\gamma^{-k,o}\in {\cal U}_{1,n}^{-k}$. 

  \end{definition}

%
%

Decentralized PbP optimality (\ref{pbp}) is 
 analogous to that of  Nash equilibrium strategies, but instead there is only one payoff common to all  strategies. 

\section{DP Equations and  Information States}
\label{discrete}
In   Lemma~\ref{lemma:payoff-pbp}   we express    payoff  $J_{n} (\gamma^{k}, \gamma^{-k,o})$ of   strategy $\gamma^k \in {\cal U}_{1,n}^k$,   w.r.t.  the \^a posteriori  PMs of  nonlinear filtering of estimating the  extended state process  $\big\{X_{t},   \Lambda_t^{-k}|t \in T_+^n \big\}$ from $\big\{I_t^k| t \in T_+^n\big\}$, $\forall k \in {\mathbb Z}_+^K$.

\begin{lemma}(Payoff of Decentralized PbP Optimality)\\
\label{lemma:payoff-pbp}
For each $k \in {\mathbb Z}_+^K$, define the  \^a posteriori  PM conditional on the realizations,  $I_t^k=\{\Delta_t^{(K)}, \Lambda_t^k\}=i_t^k=\{\delta_t^{(K)}, \lambda_t^k\}$ by\footnote{Notation $\Xi_t^k[i_t^k](dx_t, d\lambda_{t}^{-k})$ is often used  as a reminder that this is a measurable function of the data $i_t^k$.} 
\begin{align}
& \Xi_t^k[i_t^k](dx_t, d\lambda_{t}^{-k})\tri {\mathbb P}^{{\gamma^k, \gamma^{-k}}} \big\{X_{t} \in dx_t, \Lambda_t^{-k} \in d\lambda_t^{-k} \big| i_t^{k}\big\}\nonumber \\
 &= {\bf P}_t^{{\gamma^k, \gamma^{-k}}} (dx_t, d\lambda_{t}^{-k}\big|i_t^k), \;  \forall t \in T_+^n, \; \forall k \in {\mathbb Z}_+^K. \label{pPM}
\end{align} 
The    payoff $J_{n} (\gamma^{k}, \gamma^{-k,o})$ of  strategy $\gamma^k \in {\cal U}_{1,n}^k$  for fixed   optimal strategy  $\gamma^{-k}=\gamma^{-k,o}\in {\cal U}_{1,n}^{-k}$ is expressed as, 
\begin{align}
&J_{n} (\gamma^{k}, \gamma^{-k,o}) ={\mathbb E}^{^{\gamma^k, \gamma^{-k,o}}} \Big\{ \sum_{t=1}^{n} \ell(t, X_t,\gamma_t^k,  \gamma_t^{-k,o})  \Big\}\nonumber  \\
&=   {\mathbb E}^{^{\gamma^k, \gamma^{-k,o}}} \Big\{ \sum_{t=1}^{n}  \int_{{\mathbb X}_t \times {\mathbb L}_{t}^{-k}}   \ell(t, x_t,\gamma_t^k(I_t^{k}),  \gamma_t^{-k,o}(I_t^{-k}))\nonumber \\
&\hst .{\bf P}_t^{{\gamma^k, \gamma^{-k,o}}}(dx_t, d\lambda_{t}^{-k}\big|I_t^k) \Big\}, \hso \forall k \in {\mathbb Z}_+^K. \label{tpayf-pbp}
\end{align}
\end{lemma}
\begin{proof}  By 
re-conditioning on $I_t^k$ we obtain (\ref{tpayf-pbp}).
\end{proof}

%
%
%
%
%

\ \

For each $k$,  payoff  (\ref{tpayf-pbp}) depends on {\it  private \^a posteriory PM of  strategy}    $\gamma^k $, i.e.,  $\big\{\Xi_t^k[i_t^k] \equiv \Xi_t^k[i_t^k] (dx_t, d\lambda_{t}^{-k}) \big|t \in T_+^n   \big\}$,  and  common  data $\big\{\Delta_t^{(K)}  \big|t \in T_+^n   \big\}$.

Lemma~\ref{lemma:nested} is  essential to our DP approach.

\begin{lemma}(Nested  Properties of Information Patterns)\\
\label{lemma:nested}
(1) The  nested properties  $\Delta_t^{(K)} \subseteq \Delta_{t+1}^{(K)}$,   $I_t^k\subseteq I_{t+1}^k$ hold.
\begin{align}
&\Delta_{t+1}^{(K)}=\big\{\Delta_{t}^{(K)}, Y_{t-T+1}^{(K)}, U_{t-T+1}^{(K)}\big\}, \hso \: \forall (t, k, T) \label{nested_1}  \\
&=\big\{\Delta_{t}^{(K)}, Y_{t-T+1}^{k}, U_{t-T+1}^{k}, Y_{t-T+1}^{-k}, U_{t-T+1}^{-k}\big\}\\
 &\Lambda_{t+1}^k = \big\{Y_{t-T+2, t+1}^{k}, U_{t-T+2,t}^{k}\big\}, \\
&I_{t+1}^k = \big\{Y_{1,t+1}^{k}, U_{1,t}^k,Y_{1,t-T+1}^{-k}, U_{1,t-T+1}^{-k}\big\},  \label{eq-nested-a}  \\
&=\big\{I_{t}^k,Y_{t+1}^k, U_t^k, Y_{t-T+1}^{-k}, U_{t-T+1}^{-k}  \big\} \label{eq-nested-a-1} \\
&=  \big\{\Delta_{t+1}^{(K)}, \Lambda_{t+1}^k\big\}, \hso  I_t^k=\big\{\Delta_{t}^{(K)}, \Lambda_{t}^k\big\} \subseteq I_{t+1}^k .  \label{eq-nested}
\end{align}
(2) For each $(t, k)\in T_+^n\times {\mathbb Z}_+^K$ and $T \in \{1,\ldots, n\}$  consider the actions $U_{t-T+1}^{-k}$ generated by strategies  
\begin{align}
U_{t-T+1}^{-k}=&  \gamma_{t-T+1}^{-k}(\Delta_{t-T+1}^{(K)}, \Lambda_{t-T+1}^{-k}) \\
=&\Big\{ \gamma_{t-T+1}^{j}(\Delta_{t-T+1}^{(K)}, \Lambda_{t-T+1}^{j})\Big\}_{j=1, j\neq k}^K . \label{dp_1_na-str1-old}.
\end{align}
Then  $\Delta_{t-T+1}^{(K)}\subseteq \Delta_{t}^{(K)}$  and 
 $\Lambda_{t-T+1}^{j} \subseteq \{ \Delta_{t}^{(K)},Y_{t-T+1}^{j}\}, \forall  j \in {\mathbb Z}_+^K,  j \neq k, \forall t \in T_+^n$, i.e., the 
arguments      of the  strategies, $\gamma_{t-T+1}^{-k}(\cdot)$ except $Y_{t-T+1}^{j},  \forall j \neq k$,  are 
specified by $\Delta_{t}^{(K)}=  \{Y_{1, t-T}^{(K)}, U_{1, t-T}^{(K)} \}$. 
\end{lemma}
\begin{proof} The statements are easily verified   from  the definition of  information patterns.
\end{proof}


\subsection{Private and Centralized Information  States}
\label{sect:pr-isr}
In Theorem~\ref{thm:is-pbp}, for each $k$ we present  a  recursion for    the  {\it private  \^a posteriori PM},   $\big\{\Xi_t^k[I_{t}^k]={\bf P}_{t}^{ \gamma^{k}, { \gamma^{-k,o}}} (dx_{t}, d\lambda_{t}^{-k}\big|I_{t}^k)| t \in T_+^n\big\}$.
In  Theorem~\ref{thm:is-cen}, we present    Markov  recursions for  {\it centralized \^a posteriori PMs}.
 
%
%
%



%

\begin{assumptions}(Absolute Continuity)\\
\label{ass-1}
For continuous spaces,  $S_{t+1}(dx_{t+1}|x_t,  u_{t}^{(K)})$, $Q_{t}^k(dy_t^k |x_t, u_{t-1}^{(K)})$  
have probability density functions (PDFs),   
 and   for finite spaces have   transition   probability mass functions (PMFs),   $s_{t+1}(x_{t+1}|x_t,  u_{t}^{(K)}),  q_{t}^k(y_t^k|x_t,   u_{1,t}^{(K)}), \forall t$. 
\end{assumptions}


\ \

\begin{theorem}(Recursion of Private   \^a Posteriori PMs)\\
\label{thm:is-pbp}
Suppose  
Assumptions~\ref{ass-1} hold and let 
$T \in \{2,3,\ldots, n\}$. \\For each $ k  \in {\mathbb Z}_+^K$,    
 $\{\gamma_{1,n}^k,  \gamma_{1,n}^{-k}\}\in {\cal U}_{1,n}^k\times  {\cal U}_{1,n}^{-k}$, 
  the  private 
   \^a posteriori PM,  
   $\Xi_t^k[i_t^k]\equiv {\bf P}_t^{{\gamma^k, \gamma^{-k}}} (dx_t, d\lambda_{t}^{-k}\big|I_t^k=i_t^k),   i_t^k=\{\delta_t^{(K)}, \lambda_t^k\},  \forall t \in T_+^n $  satisfies   the recursion,
 \begin{align}    
&{\bf P}_{t+1}^{{\gamma^k, \gamma^{-k}}} (dx_{t+1}, d\lambda_{t+1}^{-k}\big|i_{t+1}^k),  \forall (t,k)  \label{apost_1}\  \\
&= {\bf T}_{t+1}^{k}\big(y_{t+1}^{k}, u_t^k, \gamma_t^{-k}(\delta_t^{(K)}, \cdot),{\bf P}_{t}^{{\gamma^k, \gamma^{-k}}} (\cdot\big|i_{t}^k)\big)(dx_{t+1}, d\lambda_{t+1}^{-k}) , \nonumber \\
& {\bf P}_1^{{\gamma^k, \gamma^{-k}}}(dx_1, d\lambda_1^{-k}|i_1^k)= {\bf P}_{X_1, Y_1^{-k}|Y_1^k=y_1^k}  
\end{align}
where $u_t^k=\gamma_t^k(i_t^k)$ and  the operator ${\bf T}_{t+1}^{k}[\cdot](\cdot,\cdot)$  is    
\begin{align}
&{\bf T}_{t+1}^{k}\big(y_{t+1}^{k}, u_{t}^k,  \gamma_t^{-k}(\delta_t^{(K)}, \cdot),\Xi_{t}^{k} (\cdot)\big)(dx_{t+1}, d\lambda_{t+1}^{-k}) \tri \nonumber  \\
&\frac{\overline{\bf T}_{t+1}^k\big(y_{t+1}^{k}, u_{t}^k, \gamma_t^{-k}(\delta_t^{(K)}, \cdot),\Xi_{t}^{k} (\cdot)](dx_{t+1}, d\lambda_{t+1}^{-k}) }{\int_{{\mathbb X}_{t+1}  \times {\mathbb L}_{ t+1}^{-k}  } \overline{\bf T}_{t+1}^k[y_{t+1}^{ k}, u_{t}^k, \gamma_t^{-k}(\delta_t^{(K)},\cdot),\Xi_{t}^{k} (\cdot)\big)(dx_{t+1},d\lambda_{t+1}^{-k} )   }, \nonumber  \\
&\overline{\bf T}_{t+1}^k\big(y_{t+1}^{k}, u_{t}^k, \gamma_t^{-k}(\delta_t^{(K)}, \cdot),\Xi_{t}^{k} (\cdot)\big)(dx_{t+1},d\lambda_{t+1}^{-k}) \nonumber\\\
&\tri       \int_{{\mathbb X}_t} Q_{t+1}^{k}(dy_{t+1}^{k}|x_{t+1}, u_t^k, \gamma_t^{-k}(\delta_t^{(K)}, \lambda_t^{-k}))  \nonumber \\
&.\prod_{j=1, j\neq k}^K Q_{t+1}^{j}(dy_{t+1}^{j}|x_{t+1}, u_t^k, \gamma_t^{-k}(\delta_t^{(K)},\lambda_t^{-k})) \nonumber\\ 
&. S_{t+1}(dx_{t+1}\big|x_{t},u_{t}^k, \gamma_t^{-k}(\delta_t^{(K)}, \lambda_t^{-k})) \nonumber \\
& . \prod_{j=1, j \neq k}^K P_t^j(du_t^{j}\big| \delta_t^{(K)} , \lambda_t^j) \; \Xi_{t}^{k} [i_t^k](dx_{t}, d\lambda_{t}^{-k}) \label{apost_1a} 
\end{align}
$P_t^j(du_t^{j}\big| \delta_t^{(K)}, \lambda_t^j)  = { \mu}_{\gamma_t^{j}(\delta_t^{(K)}, \lambda_t^j)}(du_t^{j})$  (the  Dirac PM).\\
For  $T=1$,  (\ref{apost_1})-(\ref{apost_1a}) hold with  the term $\prod_{j=1, j \neq k}^K P_t^j(du_t^{j}\big| \delta_t^{(K)}, \lambda_t^j)$ removed  (i.e.,  $\lambda_{t+1}^{-k}=Y_{t+1}^{-k}$). 
\end{theorem} 
\begin{proof} This is shown using variations of   \cite{kumar-varayia:B1986}. 
\end{proof}


By Theorem~\ref{thm:is-pbp},     $\Xi_{t+1}^k[i_{t+1}^k]$, is updated  using  $\Xi_{t}^k[i_{t}^k]$,   $y_{t+1}^k$,  the value 
 of  control action  $u_t^k=\gamma_t^k(i_t^k)$ at $i_t^k$ only, and the strategies $\gamma_t^{-k}(\delta_t^{(K)}, \cdot)$ at $\delta_t^{(K)}$. This leads  Lemma~\ref{lem:ind},  a generalization of  a property of    centralized POMDP \cite{kumar-varayia:B1986}.


\begin{lemma}(Indep. of  $k$  \^a Posteriori PM on  $k$ Strategy)\\
\label{lem:ind}
 ${\bf P}_{t}^{{\gamma^k, \gamma^{-k}}} (dx_{t}, d\lambda_{t}^{-k}\big|I_{t}^k=i_{t}^k)={\bf P}_{t}^{{\gamma^{-k}}} (dx_{t}, d\lambda_{t}^{-k}\big|I_{t}^k)$,  $\forall (t,    \gamma^k), \forall k$,  i.e.,  depends on the actions $u_{1,n}^k\in {\mathbb U}_{1,n}^k$ and not   the strategies  $ \gamma^k(\cdot)\in {\cal  U}_{1,n}^k$,. 
%
%
\end{lemma}
\begin{proof} The prove is  similar to centralized POMDP
\cite[Lemma~5.10, pp.81]{kumar-varayia:B1986} (by induction).  
%
%
\end{proof}

The analog of  Lemma~\ref{lem:ind} does not hold in  studies \cite{nayyar-mahajan-teneketzis2011,nayyar-mahajan-teneketzis2013,nayyar-teneketzis2019}, because they  use    PM, ${\bf P}_{t}^{\gamma^{(K)}} (dx_{t}, d\lambda_{t}^{-k}\big|\Delta_{t}^{(K)})$ or variants, where the  conditioning is on  $\Delta_{t}^{(K)},\forall  t$.

\begin{theorem}(Recursions  of Centralized \^a Poster. PMs)\\
\label{thm:is-cen}
Suppose Assumnptions~\ref{ass-1} hold. \\
\indent (1) For any   
 $\gamma^{(K)}\in  {\cal U}_{1,n}^{(K)}$, 
the centralized \^a posteriori PM,   $ \Pi_{t}^{(K)}[\delta_t^{(K)}] \equiv \Pi_{t}^{(K)}[\delta_t^{(K)}](dx_{t-T}) \tri {\bf P}_t^{\gamma^{(K)}} (dx_{t-T}\big|\Delta_t^{(K)}=\delta_t^{(K)})$,   does not depend on  strategies $\gamma^{(K)}(\cdot)$, and  satisfies the Markov recursion, 
 \begin{align}    
&{\bf P}_{t+1} (dx_{t-T+1}\big|\delta_{t+1}^{(K)}) \hst   t=T+1, \ldots,  n \label{apost_c}\  \\
&= {\bf T}_{t+1}^{(K)}\big(y_{t-T+1}^{(K)},u_{t-T}^{(K)} ,{\bf P}_{t} (\cdot\big|\delta_{t}^{(K)})\big)(dx_{t-T+1}) , \nonumber \\
& {\bf P}_{t}(dx_{t-T}|\delta_{t}^{(K)})\big|_{t=T+1}= {\bf P}_{X_1|Y_{1}^{(K)}=y_{1}^{(K)}} , \\ 
&{\bf T}_{t+1}^{(K)}\big(y_{t-T+1}^{(K)},  u_{t-T}^{(K)},\Pi_{t}^{(K)} (\cdot)\big)(dx_{t-T+1})  \nonumber  \\
&\tri\frac{ \overline{\bf T}_{t+1}^{(K)}\big(y_{t-T+1}^{(K)}, u_{t-T}^{(K)}, \Pi_{t}^{(K)} (\cdot)\big)(dx_{t-T+1}) }  { \int_{{\mathbb X}_{t-T+1} } \overline{\bf T}_{t+1}^{(K)}\big(y_{t-T+1}^{ (K)}, u_{t-T}^{(K)}, \Pi_{t}^{(K)} (\cdot)\big)(dx_{t-T+1})  }, \nonumber  \\
&\overline{\bf T}_{t+1}^{(K)}\big(y_{t-T+1}^{(K)}, u_{t-T}^{(K)},\Pi_{t}^{(K)} (\cdot)\big)(dx_{t-T+1}) \nonumber\\
&\tri       \int_{{\mathbb X}_{t-T}} S_{t-T+1}(dx_{t-T+1}\big|x_{t-T},u_{t-T}^{(K)})\nonumber  \\
&. Q_{t-T+1}^{k}(dy_{t-T+1}^{k}|x_{t-T+1}, u_{t-T}^{(K)})  \nonumber \\
&. Q_{t-T+1}^{-k}(dy_{t-T+1}^{-k}|x_{t-T+1}, u_{t-T}^{(K)}) 
\Pi_{t}^{(K)} (dx_{t-T}) \label{ope_is} 
\end{align}
(2) For any   
 $\gamma^{(K)}\in  {\cal U}_{1,n}^{(K)}$, 
the centralized \^a posteriori PM,   $ \Theta_{t}^{\gamma^{(K)}}[\delta_t^{(K)}] \equiv \Theta_{t}^{\gamma^{(K)}}[\delta_t^{(K)}](dx_{t}, d\lambda_{t}^{(K)}) \tri {\bf P}_t^{\gamma^{(K)}} (dx_{t}, d\lambda_{t}^{(K)}\big|\Delta_t^{(K)}=\delta_t^{(K)})$, satisfies the Markov recursion, 
 \begin{align}    
&{\bf P}_{t+1}^{\gamma^{(K)}} (dx_{t+1}, d\lambda_{t+1}^{(K)}\big|\delta_{t+1}^{(K)}) \hst  \forall t \in T_+^{n-1} \label{apost_c-1}\  \\
&= {\bf T}_{t+1}^{(K)}\big(y_{t+1}^{(K)}, \gamma_{t}^{(K)},{\bf P}_{t}^{\gamma^{(K)}} (\cdot\big|\delta_{t}^{(K)})\big)(dx_{t+1},d\lambda_{t+1}^{(K)}) , \nonumber \\
& {\bf P}_{1}^{\gamma^{(K)}}(dx_{1},d\lambda_{1}^{(K)}|\Delta_1^{(K)})={\bf P}_{X_1, Y_1^{(K)}}, \\
&{\bf T}_{t+1}^{{(K)}}\big(y_{t+1}^{(K)}, \gamma_{t}^{(K)},\Theta_{t}^{\gamma^{(K)}} (\cdot)](dx_{t+1},d\lambda_{t+1}^{(K)}) \nonumber  \\
& \tri\frac{ \overline{\bf T}_{t+1}^{{(K)}}\big(y_{t+1}^{(K)}, \gamma_{t}^{(K)}, \Theta_{t}^{\gamma^{(K)}} (\cdot)\big)(dx_{t+1},d\lambda_{t+1}^{(K)}) }  { \int_{{\mathbb S}_{t+1}^{(K)} } \overline{\bf T}_{t+1}^{{(K)}}\big(y_{t+1}^{ (K)}, \gamma_{t}^{(K)},\Theta_{t}^{\gamma^{(K)}} (\cdot)\big)(dx_{t+1},d\lambda_{t+1}^{(K)})  }, \nonumber 
\\
&\overline{\bf T}_{t+1}^{{(K)}}\big(y_{t+1}^{(K)}, \gamma_{t}^{(K)},\Theta_{t}^{\gamma^{(K)}} (\cdot)\big)(dx_{t+1}, d\lambda_{t+1}^{(K)})\nonumber \\
& \tri       \int_{{\mathbb X}_{t}} Q_{t+1}^{(K)}(dy_{t+1}^{(K)}|x_{t+1}, u_{t}^{(K)}) S_{t+1}(dx_{t+1}\big|x_{t},u_{t}^{(K)})   \nonumber 
\end{align}
\begin{align}
&.  \prod_{j=1}^KP_t^j(du_t^{j}\big| \delta_t^{(K)}, \lambda_t^j)   \Theta_{t}^{\gamma^{(K)}} (dx_{t},d\lambda_{t}^{(K)})\label{ope_is-1} 
\end{align}
where $\gamma_{t}^{(K)}=\{\gamma_{t}^{k}(\delta_t^{(K)}, \lambda_t^k)\}_{k=1}^K$, 
$ P_t^j(du_t^{j}\big| \delta_t^{(K)}, \lambda_t^j)  = { \mu}_{\gamma_t^{j}(\delta_t^{(K)}, \lambda_t^j)}(du_t^{j}), \forall j$, and ${\mathbb S}_{t+1}^{(K)}\tri {\mathbb X}_{t+1}\times {\mathbb L}_{t+1}^{(K)} , \forall t$.
\end{theorem} 
\begin{proof} These  are     centralized \^a posteriori PMs, hence their derivations are  standard, i.e.,  \cite{kumar-varayia:B1986}.
\end{proof}

\ \

We note that  $\Pi_{t+1}^{(K)}[\Delta_{t+1}^{(K)}] $ depends  on actions  $  u_{t-T}^{(K)}$,  while  recursion $ \Theta_{t+1}^{\gamma^{(K)}}[\Delta_{t+1}^{(K)}]  $ depends on strategies  $\gamma^{(K)}\in {\cal U}_{t}^{(K)}$.

Our   DP approach is based on structural properties of optimal strategies based on  Definition~\ref{def:is-ss}.

\begin{definition} Semi-Separated,  Separated,  Information State  Strategies)\\ 
\label{def:is-ss}
%
For each $k \in {\mathbb Z}_+^k$,  strategies $\gamma^k \in {\cal U}_{1,n}^k$  are   called,\\
 i) {\it semi-separated denoted  by ${\cal U}_{1,n}^{k,s-sep}\subseteq {\cal U}_{1,n}^k$, }     if  $U_t^k=g_t^k(\Xi_t^k[I_t^k],\Delta_t^{(K)}, \Lambda_{t}^k)$, i.e.,   ${\cal G}_t^k =\{\Xi_t^k[I_t^k],\Delta_t^{(K)}, \Lambda_t^k\}, \forall t \in T_+^n$ is a sufficient statistic  for $\gamma^k \in {\cal U}_{1,n}^k$ ; \\
 ii)  {\it separated   denoted  by ${\cal U}_{1,n}^{k,sep}\subseteq {\cal U}_{1,n}^k$, }   if  $U_t^k=g_t^k(\Xi_t^k[I_t^k], \Pi_t^{(K)}[\Delta_t^{(K)}],\Lambda_t^k )$ or $U_t^k=g_t^k(\Xi_t^k[I_t^k], \Theta_t^{\gamma^{(K)}}[\Delta_t^{(K)}],\Lambda_t^k )$, 
  $ \forall t \in T_+^n$;\\
iii) {\it information  state     denoted  by ${\cal U}_{1,n}^{k,is}\subseteq {\cal U}_{1,n}^k$ }   if   $U_t^k=g_t^k(\Xi_t^k[I_t^k], \Pi_t^{(K)}[\Delta_t^{(K)}])$, or $U_t^k=g_t^k(\Xi_t^k[I_t^k], \Theta_t^{\gamma^{(K)}}[\Delta_t^{(K)}])$,
$\forall t \in T_+^n$.
\end{definition}

%
%

Next, we show 
$ \big\{\Xi_t^k[i_t^k]={\bf P}_t^{{\gamma^k, \gamma^{-k}}} (dx_t, d\lambda_{t}^{-k}\big|I_t^k=i_t^k)\big|t\in T_+^n\big\}$  is  Markov conditional on  $ \big\{\Delta_t^{(K)}\big|t\in T_+^n\big\}$.

\ \

\begin{lemma}(Conditional Markov Prop. of  Private  \^a Posteriori PMs)
\label{lem:markov}

\indent (1)  The process  $\{ \Xi_t^{k}[I_t^k]| t \in T_+^n\}$ is  conditionally  Markov     w.r.t. $\{ \Delta_t^{(K)} | t \in T_+^n\}$,  i.e.,  
 \begin{align}
&{\mathbb  P}^{\gamma^{k}, \gamma^{-k}}\Big\{\Xi_{t+1}^k \in  d\xi_{t+1}^k \Big| I_t^k=i_t^k, U_t^k=u_t^k\Big\} \label{ext-m}   \\
&={\bf P}_{t+1}^{\gamma^{-k}}(d\xi_{t+1}^k \Big| \Xi_t^k=\xi_t^k, \Delta_t^{(K)}=\delta_t^{(K)}, U_t^k=u_t^k),\forall t, k\nonumber
\end{align}
\indent (2) 
 For separated strategies   $\gamma^{(K)}\in {\cal U}_{1,n}^{(K),sep}$  the process  $\{ \Xi_t^{k}[I_t^k] | t \in T_+^n\}$ is conditionally Markov w.r.t..   $\{ \Pi_t^{{(K)}}[\Delta_t^{(K)}] | t \in T_+^n\}$ or  $\{ \Theta_t^{\gamma^{(K)}}[\Delta_t^{(K)}] | t \in T_+^n\}$, i.e., (\ref{ext-m}) holds with    $\Delta_t^{(K)}$ replaced by $\Pi_t^{(K)}$ or $\Theta_t^{{(K)}}$.
\end{lemma}
\begin{proof} This is easy shown using the recursions.
\end{proof}

\subsection{Generalized   DP Equations, Private  Information States, Semi-Separated Strategies}
\label{sect:DP-VP}
In this section
we give \\
1) the general  necessary and sufficient conditions for PbP optimality using generalized  DP equations  for strategies  $\gamma^{(K)} \in {\cal U}_{1,n}^{(K)}$, and \\ 
2) the {\it structural property} that optimal strategies occur in the subset of semi-separated  strategies ${\cal U}_{1,n}^{(K),s-sep}\subseteq 
{\cal U}_{1,n}^{(K)}$.


%

\ \

\begin{definition}(Decentralized PbP Value Processes)\\
\label{def:ctg-nm}
Consider  PbP  payoff $J_{n} (\gamma^{k}, \gamma^{-k,o})$ of  strategy $\gamma^k \in {\cal U}_{1,n}^k$  for fixed   optimal  $\gamma^{-k}=\gamma^{-k,o}\in {\cal U}_{1,n}^{-k}$,  $\forall k \in {\mathbb Z}_+^K$.\\
Define the value process or  optimal cost-to-go ${\cal V}_{t}^{\gamma^{k,o},\gamma^{-k,o}}(\cdot): {\mathbb I}_t^{k}\rar [0,\infty)$, over  $\{t,t+1, \ldots, n\}$  of  strategy $\gamma^k \in {\cal U}_{1,n}^k$,  when the optimal strategy $\gamma^{k,o}\in {\cal U}_{1,t-1}^{k}$ is  used over $\{1,2, \ldots, t-1\}$, conditioned  on any  $I_t^k=\{\Delta_t^{(K)}, \Lambda_t^k\}=i_t^k=\{\delta_t^{(K)}, \lambda_t^k\}$   is   
\begin{align}
 & {\cal V}_{t}^{\gamma^{k,o},\gamma^{-k,o}}(i_t^k) \tri   \inf_{  \gamma^{k}\in {\cal U}_{t,n}^k}  {J}_{t,n}^{\gamma^{k},\gamma^{-k,o}}(i_t^k),   \;    \forall t,
  \label{opt-ctg-g}\\
&{J}_{t,n}^{\gamma^{k},\gamma^{-k,o}}(i_t^k) 
 \tri   {\mathbb E}^{^{\gamma^k, \gamma^{-k,o}}} \Big\{ \nonumber \\
 &\hso \sum_{j=t}^{n} \ell(j, X_j,\gamma_j^k(I_j^k),  \gamma_j^{-k,o}(I_j^{-k})) \Big|i_t^k \Big\}, \;   \forall k .    \label{opt-ctg-pbp-tc} 
\end{align}
\end{definition}

Theorem~\ref{thm:vp} gives necessary and sufficient conditions for PbP optimality.

\ \

\begin{theorem}(Decentralized generalized DP Equations-Necessary and Sufficient Conditions)\\
\label{thm:vp}
Consider  
  the value process of  Definition~\ref{def:ctg-nm}.\\
(1) {\it Necessary Conditions for PbP Optimality.}  For each $ k \in {\mathbb Z}_+^K$,  suppose a PbP optimal $\gamma^{(K),o} \in {\cal U}_{1,n}^{(K)}$ exist. \\
(1.1) 
${\cal V}_t^{\gamma^{k,o}, \gamma^{-k,o}}(i_t^k)=  {\cal V}_t^{ \gamma^{-k,o}}(i_t^k), \forall i_t^k,  \forall  \gamma^{k,o}  \in {\cal U}_{1,n}^k, \forall  t\in T_+^n $ (i.e., is independent  of $\gamma^{k,o}$) 
 and  satisfies   the   DP   equations,  (\ref{dp_1_nnn_1}),   (\ref{dp_1_na}),  for all $I_t^k=i_t^k$, 
\begin{align}
&{\cal V}_n^{\gamma^{-k,o}}(i_n^k)=\inf_{u_n^k \in {\mathbb U}_n^k}{\mathbb E}^{ \gamma^{-k,o}}   \Big\{   \ell(n, X_n,u_n^k,   \nonumber \\
&\gamma_n^{-k,o}(\delta_n^{(K)}, \Lambda_n^{-k}))   \Big| i_n^k, u_n^k  \Big\},\hst  \forall k \in {\mathbb Z}_+^K \label{dp_1_nnn}\\
&= \inf_{u_n^k \in {\mathbb U}_n^k}\int_{{\mathbb X}_n  \times {\mathbb L}_{n}^{-k} }     \ell(n, x_n,u_n^k,  \gamma_n^{-k,o}(\delta_n^{(K)}, \lambda_n^{-k})) \nonumber \\
&\hst . \xi_n^k[i_n^k] (dx_{n}, d\lambda_n^{-k}),      \label{dp_1_nnn_1}  \\
&{\cal V}_t^{\gamma^{-k,o}}(i_t^k) 
 = \inf_{u_t^k \in {\mathbb U}_t^k}  {\mathbb E}^{ \gamma^{-k,o}}   \Big\{    \ell(t, X_t,u_t^k,  \gamma_t^{-k,o}(\delta_t^{(K)}, \Lambda_t^{-k}))  \nonumber\\
 &+{\cal V}_{t+1}^{\gamma^{-k,o}}(I_{t+1}^k ) \Big| i_t^k,  u_t^k \Big\}, \; \forall t \in T_+^{n-1}, k \in {\mathbb Z}_+^K \label{dp_2_nm}   \\
 &=  \inf_{u_t^k \in {\mathbb U}_t^k}  \Big\{  \int_{{\mathbb X}_t \times {\mathbb L}_{t}^{-k}}    \ell(t, x_t,u_t^k,  \gamma_t^{-k,o}(\delta_t^{(K)}, \lambda_t^{-k}))  \nonumber \\
 &.  \xi_t^k [i_t^k](dx_{t}, d\lambda_{t}^{-k})  +  \int_{ {\mathbb Y}_{t+1}^k \times {\mathbb X}_{t, t+1}   \times {\mathbb L}_{t}^{-k}    } {\cal V}_{t+1}^{\gamma^{-k,o}}(i_{t}^k,  y_{t+1}^k, u_t^k,\nonumber\\
 & y_{t-T+1}^{-k},  u_{t-T+1}^{-k})   {\bf P}_{t+1}^{ \gamma^{-k,o}}(dy_{t+1}^k, d\lambda_{t}^{-k}, dx_{t}, dx_{t+1}\big|   i_t^k, u_t^k)\Big\} \label{dp_1_na}
\end{align}
where in  (\ref{dp_1_na}) 
 the   conditional  PM and $u_{t-T+1}^{-k}$ are
\begin{align}
&{\bf P}_{t+1}^{ \gamma^{-k,o}}( dy_{t+1}^k, d\lambda_t^{-k},dx_{t}, dx_{t+1} \big|i_t^k, u_t^k) \label{Mar_REC_1-nn-11}  \\
&={\bf P}_{t+1}^{ \gamma^{-k,o}}( dy_{t+1}^k, d\lambda_t^{-k},dx_{t}, dx_{t+1} \big| \xi_t^k, \delta_t^{(K)},  u_t^k) \label{Mar_REC_1-nn-aaa}\\
&= Q_{t+1}^k(dy_{t+1}^k\big|    x_{t+1}, u_t^k, \gamma_t^{-k,o}(\delta_t^{(K)}, \lambda_t^{-k}))\label{Mar_REC_1-nn} \\
&.{ S}_{t+1}(dx_{t+1}\big|   x_t, u_t^k,\gamma_t^{-k,o}(\delta_t^{(K)}, \lambda_t^{-k})) \xi_t^k[i_t^k](dx_{t}, d\lambda_{t}^{-k}), \nonumber \\
 &u_{t-T+1}^{-k}
 =\Big\{ \gamma_{t-T+1}^{j,o}(\delta_{t-T+1}^{(K)}, \lambda_{t-T+1}^{j})\Big\}_{j=1,j \neq k}^K, \label{dp_1_na-str1}
\end{align}
and the  argument  $\{\delta_{t-T+1}^{(K)},\lambda_{t-T+1}^{j}\}$    of  $\gamma_{t-T+1}^{j,o}(\cdot)$, 
  except  $y_{t-T+1}^{j}$,   is specified by $\delta_{t}^{(K)}$ (see Lemma~\ref{lemma:nested}.(2)).  
\\
(1.2) For each $k$, the infimum in the DP eqns occurs 
in the set of semi-separated strategies, $\gamma^{k,o} \in  {\cal U}_{1,n}^{k,s-sep}\subseteq {\cal U}_{1,n}^{k}$,   i.e., 
 $u_n^{k,o}=\gamma_n^{k,o}(\xi_n^k[i_n^k], \delta_n^{(K)})\in {\mathbb U}_n^k,   u_t^{k,o}=\gamma_t^{k,o}(\xi_t^k[i_t^k], \delta_t^{(K)}, \lambda_t^k)\in {\mathbb U}_t^k, \forall  t \in T_+^{n-1}$. \\
(2) {\it Verification-Sufficient  Conditions for PbP Optimality.} \\
(2.1) For each  $k \in {\mathbb Z}_+^K$, if the value process  ${\cal V}_t^{ \gamma^{-k,o}}(\cdot),  \forall t \in T_+^n$  satisfies the DP equation  of part (1),  then  (almost surely), 
\begin{align}
 {\cal V}_t^{\gamma^{-k,o}}(I_t^k)    \leq J_{t,n}^{\gamma^{k}, \gamma^{-k,o}}(I_t^k) ,  \;\forall  \gamma^k \in {\cal U}_{1,n}^k, \; \forall (t,k)  \label{in_1_g} 
\end{align}
and the resulting $\gamma^{k,o} \in {\cal U}_{1,n}^k, \; \forall k \in {\mathbb Z}_+^K$ is PbP optimal.\\
(2.2) Suppose for each $k \in {\mathbb Z}_{+}^K$, 
$\gamma^{k,o}(\cdot)$ is a strategy $  u_t^{k,o}=\gamma_t^{k,o}(\xi_t^k[i_t^k], \delta_t^{(K)}, \lambda_t^k)\in {\mathbb U}_t^k$,  i.e., $\gamma^{k,o} \in   {\cal U}_{1,n}^{k,s-sep}$, 
such that,  for all $\{\xi_t^k[i_t^k], \delta_t^{(K)}, \lambda_t^k\}$ achieves   the infimum in the DP eqns of part (1)  for $t=1, \ldots, n$. \\
Then $\gamma^{k,o} \in  {\cal U}_{1,n}^{k}$ is PbP optimal and ${\cal V}_t^{\gamma^{-k,o}}(I_t^k) = J_{t,n}^{\gamma^{k,o}, \gamma^{-k,o}}(I_t^k), \forall (t, k)$  (almost surely).
\end{theorem}
\begin{proof} The proof  is a  variation  of the analogous proof  of centralized POMDPs \cite{kumar-varayia:B1986}.
\end{proof}

 
 \ \
 
The decentralized  DP equations 
of  Theorem~\ref{thm:vp} did not  appeared in the past literature, i.e.,  \cite{sandell-athans1974,yoshikawa1975,varaiya-walrand1978,nayyar-mahajan-teneketzis2011,nayyar-mahajan-teneketzis2013,nayyar-teneketzis2019}, because previous studies considered a single cost-to-go conditioned on  the shared information  $\{\Delta_t^{(K)}|t \in T_+^n\}$.

Corollary~\ref{cor:vp-mn} is a simplification 
of Theorem~\ref{thm:vp}.

\ \

\begin{corollary}(Decentralized DP Equations  with Private Information  State and Semi-Separated Strategies)\\
\label{cor:vp-mn}
Consider Theorem~\ref{thm:vp}. 
 The identity holds, 
\begin{align}
 {\cal V}_{t}^{\gamma^{-k,o}}(i_t^k) =V_t^{\gamma^{-k,o}}(\xi_t^k,\delta_t^{(K)},\lambda_t^k ), \;  \forall t \in T_+^n, \forall k \in {\mathbb Z}_+^K    \nonumber    
\end{align}
where $V_t^{\gamma^{-k,o}}(\xi_t^k,\delta_t^{(K)},\lambda_t^k )$ satisfies the 
  DP   equations: 
\begin{align}
&{ V}_n^{\gamma^{-k,o}}(\xi_n^k,\delta_n^{(K)}, \lambda_n^k  )=\inf_{u_n^k \in {\mathbb U}_n^k}{\mathbb E}^{ \gamma^{-k,o}}   \Big\{ \label{dp_1_nnn-mn}   \\
& \ell(n, X_n,u_n^k,  \gamma_n^{-k,o}(\delta_n^{(K)}, \Lambda_n^{-k}))   \Big| \xi_n^k,\delta_n^{(K)}, \lambda_n^k  , u_n^k \Big\}, \forall k \in {\mathbb Z}_+^K \nonumber \\
&= \inf_{u_n^k \in {\mathbb U}_n^k}\int_{{\mathbb X}_n  \times {\mathbb L}_{n}^{-k} }     \ell(n, x_n,u_n^k,  \gamma_n^{-k,o}(\delta_n^{(K)}, \lambda_n^{-k})) \label{dp_1_nnn_1-mn}\\
&\hst . \xi_{n}^{k} (dx_{n}, d\lambda_n^{-k})[i_n^k],      \nonumber   \\
&{ V}_t^{\gamma^{-k,o}}(\xi_t^k,\delta_t^{(K)},\lambda_t^k  )  = \inf_{u_t^k \in {\mathbb U}_t^k}  {\mathbb E}^{ \gamma^{-k,o}}   \Big\{ \label{dp_2_mn}   \\
&   \ell(t, X_t,u_t^k,  \gamma_t^{-k,o}(\delta_t^{(K)}, \Lambda_t^{-k}))  +{ V}_{t+1}^{\gamma^{-k,o}}(\Xi_{t+1}^k, \Delta_{t+1}^{(K)}, \Lambda_{t+1}^k  )\nonumber\\
 & \Big| \xi_t^k,\delta_t^{(K)}, \lambda_t^k , U_t^k \Big\}, \hso  \forall t \in T_+^{n-1},  \forall k \in {\mathbb Z}_+^K \label{dp_2_mn}   \\
 &=  \inf_{u_t^k \in {\mathbb U}_t^k}  \Big\{  \int_{{\mathbb X}_t \times {\mathbb L}_{t}^{-k}}    \ell(t, x_t,u_t^k,  \gamma_t^{-k,o}(\delta_t^{(K)}, \lambda_t^{-k}))  \nonumber \\
 &.  \xi_{t}^k (dx_{t}, d\lambda_{t}^{-k})[i_t^k]  +  \int_{ {\mathbb Y}_{t+1}^k \times {\mathbb X}_{t, t+1}   \times {\mathbb L}_{t}^{-k}    }  {V}_{t+1}^{\gamma^{-k,o}}({\bf T}_{t+1}^{k}[y_{t+1}^{k},  \nonumber  \\
 &u_t^k, \gamma_t^{-k}(\delta_t^{(K)}, \cdot),    \xi_t^k(\cdot)], \delta_t^{(K)}, \lambda_t^k ,   y_{t+1}^k, u_t^k,    y_{t-T+1}^{-k}, u_{t-T+1}^{-k}) \nonumber  \\
 &. {\bf P}_{t+1}^{ \gamma^{-k,o}}(dy_{t+1}^k, d\lambda_{t}^{-k}, dx_{t}, dx_{t+1}\Big|   \xi_t^k, \delta_t^{(K)}, \lambda_t^k, u_t^k) \Big\}  \label{dp_1_na_m}
\end{align}
where   
 the PM of the last RHS term is
%
given in  Theorem~\ref{thm:vp}).\\
\end{corollary}
\begin{proof} This is a  consequence of Theorem~\ref{thm:vp}.
\end{proof}

\ \

\begin{remark} (On Corollary~\ref{cor:vp-mn})\\
\indent  In Corollary~\ref{cor:vp-mn},  
 $\gamma^{k,o}\in {\cal U}_{1,n}^{k,s-sep}$,  but  expansive w.r.t.  $\Delta_t^{(K)},\forall  t$.  The DP equations  can be   further simplified,  if instead of $\big\{\Delta_t^{(K)} | t  \in T_+^{n-1}\big\}$  we use  information states $\big\{ \Pi_{t}^{(K)}[\Delta_t^{(K)}]  | t  \in T_+^{n-1}\big\}$  or  $\big\{ \Theta_{t}^{\gamma^{(K)}}[\Delta_t^{(K)}]  | t  \in T_+^{n-1}\big\}$, and restrict  $\gamma^k \in {\cal U}_{1,n}^{k}$ to the   separated strategies ( Definition~\ref{def:is-ss})  $\gamma^k\in {\cal U}_{1,n}^{k,sep}$.
We address this  in  the next section. 
\end{remark}

\subsection{DP Equations, Private and Centralized  Information States, Separated  and Information state  Strategies}
\label{sect:DP-IS-SEP}
Theorem~\ref{thm:vp-mn-sep} and Theorem~\ref{thm:vp-mn-is},   settle the long standing open problem, first discussed by Witsenhausen in \cite[Assertion 8, pp.1562]{witsenhausen1971}. We make use of the  centralized \^a posteriori PMs, Theorem~\ref{thm:is-cen},  to derive    2   additional simplified  DP equations,  using separable strategies    ${\cal U}_{1,n}^{k,sep}$ and information state strategies ${\cal U}_{1,n}^{k,is}$,  which are not expansive w.r.t.  $\Delta_t^{(K)},\forall  t$.

\ \

\begin{theorem}(DP Eqns with Private and Centralized  Information States-Necessary Conditions and Verification)\\
\label{thm:vp-mn-sep}
Suppose   Assumptions~\ref{ass-1} hold.  
 Restrict  the strategies to  separated,   $ {\cal U}_{1,n}^{k,sep}\subseteq {\cal U}_{1,n}^{k,s-sep}\subseteq{\cal U}_{1,n}^{k}$,   $\gamma_t^k=\gamma_t^{k}(\Xi_t^k[I_t^k], \Pi_t^{(K)}[\Delta_{t}^{(K)}], \Lambda_t^k), \forall t, k$. 

 (1) For each $ k \in {\mathbb Z}_+^K$ suppose a PbP optimal $\gamma^{k,o} \in {\cal U}_{t,n}^{k,sep}$ exists. The identity holds, $ \forall (\xi_t^k,\delta_t^{(K)},\pi_t^{{(K)}},\lambda_t^k )$, 
\begin{align}
  V_t^{\gamma^{-k,o}}(\xi_t^k,\delta_t^{(K)},\lambda_t^k )={ V}_t^{\gamma^{-k,o},sep}(\xi_t^k,\pi_t^{(K)},\lambda_t^k  ), \;  \forall (t,k)    \nonumber    
\end{align}
 where $V_t^{ \gamma^{-k,o},sep}(\cdot)$  
   satisfies    the following  DP   equations.  
\begin{align}
&{ V}_n^{\gamma^{-k,o},sep}(\xi_n^k,\pi_n^{(K)}, \lambda_n^k  )=\inf_{u_n^k \in {\mathbb U}_n^k}{\mathbb E}^{ \gamma^{-k,o}}   \Big\{ \ell(n, X_n,u_n^k,   \nonumber  \\
&  \gamma_n^{-k,o}(\Xi_n^{-k}, \pi_n^{(K)}, \Lambda_n^{-k}))   \Big| \xi_n^k,\pi_n^{(K)}, \lambda_n^k  , u_n^k \Big\}, \forall k \label{dp_1_nnn-mn-is-s} \\
&= \inf_{u_n^k \in {\mathbb U}_n^k}\int_{{\mathbb X}_n  \times {\mathbb L}_{n}^{-k} }     \ell(n, x_n,u_n^k,  \gamma_n^{-k,o}(\xi_n^{-k}, \pi_n^{(K)}, \lambda_n^{-k}))\nonumber \\
&\hst . \xi_{n}^{k}[i_n^k] (dx_{n}, d\lambda_n^{-k}),        \label{dp_1_nnn_1-mn-is-s} \\
&{ V}_t^{\gamma^{-k,o},sep}(\xi_t^k,\pi_t^{(K)},\lambda_t^k  )   \hst \hso     \forall t \in T_+^{n-1}, \; \forall k  \label{dp_2_mn-is-s}\\
& = \inf_{u_t^k \in {\mathbb U}_t^k}  {\mathbb E}^{ \gamma^{-k,o}}   \Big\{ \ell(t, X_t,u_t^k,  \gamma_t^{-k,o}(\Xi_t^{-k}, \pi_t^{(K)}, \Lambda_t^{-k}))  \nonumber \\
&+{ V}_{t+1}^{\gamma^{-k,o}, sep}(\Xi_{t+1}^k, \Pi_{t+1}^{(K)}, \Lambda_{t+1}^k )\Big| \xi_t^k,\pi_t^{(K)}, \lambda_t^k , u_t^k, u_{t-T}^{(K)} \Big\}\nonumber    \\
 &=  \inf_{u_t^k \in {\mathbb U}_t^k}  \Big\{  \int_{{\mathbb X}_t \times {\mathbb L}_{t}^{-k}}    \ell(t, x_t,u_t^k,  \gamma_t^{-k,o}(\xi_t^{-k}, \pi_t^{(K)}, \lambda_t^{-k}))  \nonumber \\
 &.  \xi_{t}^k[i_t^k] (dx_{t}, d\lambda_{t}^{-k})  +  \int_{ {\mathbb Y}_{t+1}^k \times {\mathbb X}_{t, t+1}   \times {\mathbb L}_{t}^{-k}    }   {V}_{t+1}^{\gamma^{-k,o},sep}({\bf T}_{t+1}^{k}\big(\nonumber  \\
 &y_{t+1}^{k},  u_t^k, \gamma_t^{-k,o}(\xi_t^{-k}, \pi_t^{(K)}, \cdot),    \xi_t^k(\cdot)\big), {\bf T}_{t+1}^{(K)}\big(y_{t-T+1}^{k},    y_{t-T+1}^{-k}, \nonumber  \\
 &   u_{t-T}^{(K)}, \theta_{t}^{(K)} (\cdot)\big),     y_{t+1}^k, u_t^k,  y_{t-T+2,t}^{k},  u_{t-T+2,t-1}^{k}) \nonumber  \\
&. {\bf P}_{t+1}^{ \gamma^{-k,o}}(dy_{t+1}^{k}, d\lambda_{t}^{-k}, dx_{t}, dx_{t+1}     \Big|   \xi_t^k, \pi_t^{(K)}, \lambda_t^k , u_t^k, u_{t-T}^{(K)}) \Big\}  \label{dp_1_na_m-is-s}
\end{align}
where  
 the   conditional  PM of the last right hand side term is given in  Theorem~\ref{thm:vp} with $\delta_t^{(K)}$ replaced by $\pi_t^{(K)}$.

(2) For each $ k \in {\mathbb Z}_+^K$ suppose   the DP equations (\ref{dp_1_nnn-mn-is-s})-(\ref{dp_1_na_m-is-s}) hold. Then 
the following inequalities hold, 
\begin{align}
&{V}_n^{\gamma^{-k,o},sep}(\Xi_n^k,\Pi_n^{(K)}, \Lambda_n^k ) \label{in_1-a-is-s} \\ 
&\leq J_{n,n}^{\gamma^{k}, \gamma^{-k,o}}(I_n^k) , \;   \forall \gamma^k \in {\cal U}_{1,n}^{k,sep},  \forall k \in {\mathbb Z}_+^K, \nonumber  \\
&{ V}_t^{\gamma^{-k,o},sep}(\Xi_t^k,\Pi_t^{(K)},\Lambda_t^k )  \label{in_1-is-s}\\
& \leq J_{t,n}^{\gamma^{k}, \gamma^{-k,o}}(I_t^k) , \; \forall t \in T_+^{n-1}, \; \forall \gamma^k \in {\cal U}_{1,n}^{k,sep}, \; \forall k \in {\mathbb Z}_+^K. \nonumber 
\end{align}
where ${J}_{t,n}^{\gamma^{k},\gamma^{-k,o}}(I_t^k)$ is defined  by  (\ref{opt-ctg-pbp-tc})  w.r.t. $\{\gamma^{k},\gamma^{-k,o}\} \in {\cal U}_{1,n}^{k,sep}\times {\cal U}_{1,n}^{-k,sep}$.

(3)  Given the optimal separated strategies  $\gamma^{-k}=\gamma^{-k,o}\in {\cal U}_{1,n}^{-k, sep}$,   let $\gamma^{k,o} \in {\cal U}_{1,n}^{k,sep}$ be a separated strategy such that for all $\{\xi_t^k,\pi_t^{(K)},\lambda_{t}^k\}$, strategy $\gamma_t^{k,o}(\xi_t^k,\pi_t^{(K)},\lambda_{t}^k)$ achieves the infimum in DP eqns   (\ref{dp_1_nnn-mn-is-s})-(\ref{dp_1_na_m-is-s}) .   Then $\gamma^{k,o}(\xi_t^k,\pi_t^{(K)},\lambda_t^k )$ is optimal and $\forall k \in {\mathbb Z}_+^K$,
\begin{align}
&{ V}_n^{\gamma^{-k,o}}(\Xi_n^k,\Pi_t^{(K)}, \Lambda_t^k)  = J_{n,n}^{\gamma^{k,o}, \gamma^{-k,o}}(I_n^k)-a.s., \label{in_2-a}\\
&{V}_t^{\gamma^{-k,o}}(\Xi_t^k,\Pi_t^{(K)},\Lambda_{t}^k) \nonumber \\
 & = J_{t,n}^{\gamma^{k,o}, \gamma^{-k,o}}(I_t^k)-a.s., \forall t \in T_+^{n-1}. \label{in_2}
\end{align}
where ${J}_{t,n}^{\gamma^{k,o},\gamma^{-k,o}}(I_t^k)$ is defined  as in (2). 

\end{theorem}
\begin{proof} Similar to Theorem~\ref{thm:is-pbp}.
\end{proof} 

\ \

In Theorem~\ref{thm:vp-mn-is}, we use $\big\{ \Theta_{t}^{\gamma^{(K)}}[\Delta_t^{(K)}]  | t  \in T_+^{n-1}\big\}$ instead of  $\big\{ \Pi_{t}^{(K)}[\Delta_t^{(K)}]  | t  \in T_+^{n-1}\big\}$,  
and   we show
the simplification, $ {\cal V}_{t}^{\gamma^{-k,o}}(i_t^k)=V_t^{\gamma^{-k,o},sep}(\xi_t^k,\theta_t^{\gamma^{(K)}},\lambda_t^k  ) =V_t^{\gamma^{-k,o},is}(\xi_t^k,\theta_t^{\gamma^{(K)}})
  $, i.e.,  
  depend on $\lambda_t^k$ through $\xi_t^k,    \forall t$.

\ \

\begin{theorem}(DP Equations with Private and Centralized  Information States Depended on Strategies)\\
\label{thm:vp-mn-is}
Suppose   Assumptions~\ref{ass-1} hold. 
Restrict  the strategies to    $ {\cal U}_{1,n}^{k,sep}\subseteq {\cal U}_{1,n}^{k}$,   $\gamma_t^k=\gamma_t^{k}(\Xi_t^k[I_t^k], \Theta_t^{\gamma^{(K)}}[\Delta_{t}^{(K)}], \Lambda_t^k), \forall t, k$. 
 The following hold.  
  
  (1)  For each $ k \in {\mathbb Z}_+^K$ suppose a PbP optimal $\gamma^{k,o} \in {\cal U}_{1, n}^{k,sep}$ exists. The 
    identity holds, $ \forall (i_t^k, \xi_t^k,\delta^{(K)}, \theta_t^{\gamma^{(K)}},\lambda_t^k )$, 
\begin{align*}   
  &  {\cal V}_t^{\gamma^{-k,o}}(i_t^k )=  {V}_t^{\gamma^{-k,o},sep}(\xi_t^k,\theta_t^{\gamma^{(K)}},\lambda_t^k  ) =V_t^{\gamma^{-k,o},is}(\xi_t^k,\theta_t^{\gamma^{(K)}})
\end{align*}    
where  
  $V_t^{\gamma^{-k,o},is}(\cdot)$ 
 satisfies the DP equations, 
\begin{align}
&{ V}_n^{\gamma^{-k,o},is}(\xi_n^k,\theta_n^{\gamma^{(K)}})= \inf_{u_n^k \in {\mathbb U}_n^k}\int_{{\mathbb X}_n  \times {\mathbb L}_{n}^{-k} } \ell(n, x_n, \nonumber   \\
& u_n^k,   \gamma_n^{-k,o}(\xi_n^{-k},\theta_n^{\gamma^{(K)}}))  \xi_{n}^{k}[i_n^k] (dx_{n}, d\lambda_n^{-k}),      \label{dp_1_nnn_1-mn-is} \\
&{ V}_t^{\gamma^{-k,o},is}(\xi_t^k,\theta_t^{\gamma^{(K)}} )  \label{dp_2_mn-is}   \\
 &=  \inf_{u_t^k \in {\mathbb U}_t^k}  \Big\{  \int_{{\mathbb X}_t \times {\mathbb L}_{t}^{-k}}    \ell(t, x_t,u_t^k,  \gamma_t^{-k,o}(\xi_t^{-k},\theta_t^{\gamma^{(K)}}))  \nonumber \\
 &.  \xi_{t}^k[i_t^k]  (dx_{t}, d\lambda_{t}^{-k}) +  \int_{ {\mathbb Y}_{t+1}^k \times {\mathbb X}_{t, t+1}   \times {\mathbb L}_{t}^{-k}    } \Big[{V}_{t+1}^{\gamma^{-k,o},is}({\bf T}_{t+1}^{k}\big( \nonumber  \\
 &y_{t+1}^{k}, u_t^k,\gamma_t^{-k,o}(\xi_t^{-k}, \theta_t^{\gamma^{(K)}}),    \xi_t^k(\cdot)\big), {\bf T}_{t+1}^{{(K)}}\big(y_{t+1}^{(K)}, u_{t}^{k}, \nonumber  \\ 
 & \gamma_{t}^{-k,o}(\xi_t^{-k},\theta_t^{\gamma^{(K)}}), \theta_{t}^{\gamma^{(K)}} (\cdot)\big)) \nonumber  \\
 &. {\bf P}_{t+1}^{ \gamma^{-k,o}}(dy_{t+1}^{(K)}, d\lambda_{t}^{-k}, dx_{t}, dx_{t+1}\big|   \xi_t^k, \theta_t^{\gamma^{(K)}}, u_t^k) \Big]  \Big\} \label{dp_1_na_m-is}
\end{align}
where  
 the   conditional  PM of the last RHS term is specified.
Moreover,  if  the infimum in the DP equations exists, then the optimal strategy occurs in the set of  strategies, $\gamma^{k,o} \in  {\cal U}_{1,n}^{k, is}$, i..e., ${\cal G}_t^k\tri \{\Xi_t^k,\Theta_t^{\gamma^{(K)}}\}$ is a sufficient statistic for $U_t^k$, $\forall t$.

(2)   The statements of Theorem~\ref{thm:vp-mn-sep}.(2), (3)  hold w.r.t.    ${ V}_t^{\gamma^{-k,o},is}(\xi_t^k,\theta_t^{(K)} ), \forall (t, k)  $. \\
\end{theorem}
\begin{proof} 
 (1)    This is  shown   using  induction. 
  (2) This follows from the statements of Theorem~\ref{thm:vp-mn-sep}. 
\end{proof}

\ \

\section{conclusion}
The main contribution of this paper is  the extension of the classical DP approach of POMDP and its properties  to  decentralized  Markovian networks,   with multiple agents each    assigned   information patterns with    private   and   delayed shared components. Necessary and sufficient conditions are derived using different generalized DP equations and information states, based on  the concept of  Person-by-Person  (PbP)  optimality of team theory.  A separation principle is shown based private and centralized    information states, which   are   sufficient statistic for  the agent's  strategies.

\bibliographystyle{IEEEtran}

\bibliography{Bibliography_Decentralized_fixed}

\end{document}